\documentclass[letterpaper, amsfonts, amssymb, amsmath, pra, aps, reprint, showkeys, nofootinbib]{revtex4-2}

\usepackage{array}
\usepackage{amsmath}
\usepackage{amsfonts}
\usepackage{amsthm}
\usepackage{algorithm}
\usepackage{algpseudocode}
\usepackage{bm, bbm}
\usepackage{booktabs} 
\usepackage{enumerate}
\usepackage{graphicx}
\usepackage{hyperref}
\usepackage{longtable}
\usepackage{multirow}
\usepackage{natbib}
\usepackage{physics}
\usepackage{setspace}
\usepackage{threeparttable}
\usepackage{upgreek}

\renewcommand{\mu}{\upmu}
\hypersetup{
    colorlinks = true,
    allcolors = black
}

\graphicspath{{figures/}}

\newcommand{\beginsupplement}{%
    \setcounter{table}{0}
    \renewcommand{\thetable}{S\arabic{table}}%
    \setcounter{figure}{0}
    \renewcommand{\thefigure}{S\arabic{figure}}%
    \renewcommand{\thesubsection}{\arabic{subsection}}
 }

\newtheorem{theorem}{\bf Theorem}



\begin{document}
\title{A quantum speedup algorithm for TSP based on quantum dynamic programming with very few qubits}

\author{Xujun Bai$^{1, 2, \dag}$}
\author{Yun Shang$^{1, 3, \ddag, *}$}

\affiliation{
$^1$Institute of Mathematics, Academy of Mathematics and Systems Science, Chinese Academy of Sciences, Beijing 100190, China \\
$^2$School of Mathematical Sciences, University of Chinese Academy of Sciences, Chinese Academy of Sciences, Beijing, 100049, China \\
$^3$State Key Laboratory of Mathematical Sciences, Academy of Mathematics and Systems Science, Chinese Academy of Sciences, Beijing, 100190, China \\
}

\let\thefootnote\relax\footnotetext{$^\dag$ baixujun@amss.ac.cn}
\let\thefootnote\relax\footnotetext{$^\ddag$ shangyun@amss.ac.cn}
\let\thefootnote\relax\footnotetext{$^*$ Corresponding authors }


\begin{abstract}
    The Traveling Salesman Problem (TSP) is a classical NP-hard problem that plays a crucial role in combinatorial optimization. In this paper, we are interested in the quantum search framework for the TSP because it has robust theoretical guarantees. However, we need to first search for all Hamiltonian cycles from a very large solution space, which greatly weakens the advantage of quantum search algorithms. To address this issue, one can first prepare a superposition state of all feasible solutions, and then amplify the amplitude of the optimal solution from it. We propose a quantum algorithm to generate the uniform superposition state of all $N$-length Hamiltonian cycles as an initial state within polynomial gate complexity based on pure quantum dynamic programming with very few ancillary qubits, which achieves exponential acceleration compared to the previous initial state preparation algorithm. As a result, we realized the theoretical minimum query complexity of quantum search algorithms for a general TSP. Compared to some algorithms that theoretically have lower query complexities but lack practical implementation solutions, our algorithm has feasible circuit implementation. Our work provides a meaningful research case on how to fully utilize the structures of specific problems to unleash the acceleration capability of the quantum search algorithms.
\end{abstract}


\maketitle

\section{Introduction}

Quantum computing has the potential to surpass classical computing and bring about a computing revolution~\cite{montanaro2016quantum, Frank2019Quantum}. There are some quantum algorithms that have confirmed quantum superiority, such as Shor's integer factorization~\cite{shor1994algorithms}, Grover's unstructured database search~\cite{grover1996fast,boyer1998tight}, HHL algorithm for linear systems of equations~\cite{PhysRevLett.103.150502}. It is important to fully utilize the structures of specific problems to design algorithms that can leverage quantum advantages.

Combination optimization is considered one of the most promising fields for quantum advantages. Quantum algorithms for solving combinatorial optimization problems mainly include the quantum gate circuit and the quantum adiabatic algorithm~\cite{2000quant.ph..1106F}. The Grover's adaptive search (GAS) algorithm is a type of quantum gate circuit algorithm that provides quadratic speedup with guaranteed success probability~\cite{nielsen2002quantum}. The quantum approximate optimization algorithm (QAOA)~\cite{Farhi2014AQA} is also a promising quantum gate circuit algorithm that evolved from the quantum adiabatic algorithm.
The quantum adiabatic algorithm and the annealed quantum algorithm~\cite{PhysRevE.70.057701, RevModPhys.80.1061} are two technologies closely related. The latter is a quantum-inspired optimization method that uses quantum tunneling to search for local minima. All three methods can use the Ising model and its extensions to solve combinatorial optimization problems. That is because the constraints and objective functions for most combinatorial optimization problems can be modeled with binary variables, which usually adopt forms such as quadratic unconstrained binary optimization (QUBO)~\cite{Lucas_2014} or higher-order unconstrained binary optimization (HOBO)~\cite{ramezani2024reducingnumberqubitsn2}. By relating binary variables to spins, these problems can be mapped to the Ising system and its extensions, and their solutions are associated with the ground states of these systems.
 
The Traveling Salesman Problem (TSP) is a well-known NP-hard combinatorial optimization problem, which requires a salesman to find routes with the lowest cost to traverse each city exactly once. Classical algorithms, such as exact methods~\cite{laporte1992traveling, chauhan2012survey}, approximation methods~\cite{christofides1976worst}, heuristics methods~\cite{helsgaun2000effective, johnson1990local}, and AI-based methods~\cite{bengio2021machine, lombardi2018boosting}, cannot effectively solve the TSP. Hence, researchers have started to use quantum algorithms to solve this problem.
Ambainis et al. combined Grover’s search by computing a partial dynamic programming table and proposed an algorithm to solve the TSP with a theoretical query complexity of $O^*(1.728^N)$, which is superior to the best classical algorithms~\cite{doi:10.1137/1.9781611975482.107}. However, implementing this hybrid algorithm using specific quantum circuits while maintaining the same complexity is a challenging problem.
To overcome the problem of high dimensionality of the Hilbert space in the physical system of the Ising model, Vargas-Calder\'{o}n et al. utilized a system of many-qudits to model the TSP and adopted variational Monte Carlo (VMC) to optimize a neural quantum state (NQS) to reach its ground state~\cite{Vargas_Calder_n_2021}.
He used graph neural networks (GNN) to learn the representation of the graph structure and built a loss function based on the QUBO model to solve the TSP. He demonstrated that QUBO-based Quantum Annealing can enhance the GNN framework to address complex combinatorial optimization problems~\cite{He_2024}.
Ramezani et al. improved QAOA by reducing the qubit count from $N^2$ to $N\log N$~\cite{ramezani2024reducingnumberqubitsn2}.
QUBO needs to introduce slack variables to deal with inequality constraints, which increases the number of qubits and operations. To overcome this problem, a method using unbalanced penalization functions to encode inequality constraints has been proposed~\cite{Monta_ez_Barrera_2024}.
Goldsmith and Day-Evans proposed a new formulation that is different from QUBO and HOBO to reduce the use of binary variables and avoid the penalty term. They showed on a quantum boson sampler that larger networks can be solved with this penalty-free formulation than with formulations with penalties~\cite{goldsmith2024qubohoboformulationssolving}.
Tensor networks have also been used to design quantum-inspired algorithms for solving the TSP~\cite{ali2024travelingsalesmanproblemtensor}.
The Path-Slicing Strategy was used to decompose a TSP into manageable subproblems followed by quantum optimization~\cite{liu2024quantumlocalsearchtraveling}. This hybrid quantum-classical approach is an effective attempt to address the challenging nature of TSP computing.
Goswami et al. presented an interesting single-qubit method for the TSP combined with optimal control methods. They encoded cities in a two-dimensional plane onto a Bloch sphere and then used optimal control methods to selectively create quantum superposition states to find the shortest path~\cite{goswami2024solvingtravellingsalesmanproblem}.

In this study, we are interested in the quantum search framework for the TSP because it has robust theoretical guarantees. Two realizable algorithms in this direction can serve as references.
Zhu et al. designed and improved the Hamiltonian cycle detection (HCD) oracle to determine valid cycles using quantum circuits. Combining the HCD oracle and the quantum phase estimation, they proposed a quantum algorithm for TSP based on GAS with a query complexity of $O(\sqrt{2^{\lceil\log d\rceil N}})$, where $d$ is the maximum degree of the graph~\cite{zhu2022realizablegasbasedquantumalgorithm}. And this complexity is approximately equal to $O(\sqrt{N^N})$ for $d\sim O(N)$ under the statistical significance of random instances.
Some researchers~\cite{sato2024circuitdesigntwostepquantum} realized the importance of preparing the initial state or the superposition state encoding the feasible solutions. They proposed a two-step quantum search (TSQS) algorithm based on higher-order unconstrained binary optimization (HOBO) representation. The first step amplifies the amplitude of feasible solutions to prepare a uniform superposition state, from which the second step determines the optimal solution. The query complexity of the entire algorithm is $O(\sqrt{N!})$, which is close to the minimum query complexity of quantum search algorithms for a general TSP. However, they need the query complexity of $O(\sqrt{2^{N\log N}/N!})\approx O(e^{N/2}/N^{1/4})$ (according to Sterling's approximation) to prepare the initial state.

In this article, we show that the HCD oracle is an unnecessary design because we can exactly generate the uniform superposition state of all $N$-length Hamiltonian circles (HCs) within polynomial gate complexity based on pure quantum dynamic programming with only $\log N$ ancillary qubits, which we call HC-generation algorithm. This means that we only need to find the circle with minimum weight in a smaller space composed of all HCs, resulting in a lower query complexity of $O(\sqrt{(N - 1)!})$ which is the lowest query complexity of the pure quantum search algorithm for a general $N$-node TSP. Additionally, we employed a shortcut of Quantum Fourier Transform ($QFT$) to calculate the weights of HCs, which makes it simpler to implement quantum circuits than quantum phase estimation by avoiding the construction of a unitary matrix with the eigenvalues corresponding to the weights of HCs. This shortcut was proposed in~\cite{Gilliam_2021CPBO} to construct efficient oracles for solving constrained polynomial binary optimization (CPBO) problems based on GAS, while we use the shortcut to design efficient oracles for the TSP. Finally, we compressed the total number of qubits greatly. Although our HC-generation algorithm requires $\log N$ ancillary qubits, we can alternately use the value registers to compute the weights and act as auxiliary registers by exploiting the reversibility of quantum computing so that our algorithm does not need ancillary qubits explicitly. We significantly improved the quantum search algorithm for the TSP and successfully implemented our algorithm using IBM’s Qiskit. Under the optimal number of iterations we could obtain almost 100\% accuracy on examples of 4-8 nodes TSP where sampling shots equaled 1000. For example, we only need 30 qubits to solve an eight-node TSP with six optimal solutions, and the success rate is 99.7\%.

\section{Methods}

    The classical $N$-node TSP is described as follows. Each node $i\ (0\le i \le N-1)$ has $d_i\ (2\le d_i \le N-1)$ edges connected to other nodes. There is at most one edge that connects any two nodes. 
    The solution is a Hamiltonian cycle (HC) with the smallest total weight. (The total weight refers to the sum of the weights of all edges in an HC, which will be the same in the following text.)
    The TSP graph is complete when $d_i=N-1, \forall i$. If a TSP graph is incomplete, we can complete it by adding some edges with sufficiently large weights. Moreover, we can transform the Hamiltonian cycle problem into a complete TSP in a similar manner.
    However, considering that it requires too many qubits to store large weights, we can refer to the encoding method from the literature \cite{zhu2022realizablegasbasedquantumalgorithm}, in which the authors used some techniques to encode sparse TSP, such as adjacency lists, quantum-addressed quantum registers (QAQR) and quantum-addressed classical registers (QACR). In our study, without loss of generality, for simplicity, we only consider the complete TSP directly.
    
    There are two fundamental difficulties in designing quantum algorithms for TSP based on GAS. One is the feasibility of the solution given by the quantum algorithms. The other is that the quantum algorithms require a large number of qubits, which currently cannot be achieved in the NISQ era. 
    Our main work is divided into two parts. 
    (1) Design algorithm to exactly generate the uniform superposition state of all $N$-length HCs within polynomial gate complexity based on quantum dynamic programming.
    (2) Employ a shortcut of $QFT$ to calculate the weights of HCs.
    Then, we can extract the desired solution from a uniform superposition state of all $N$-length HCs using standard Grover's procedure.
    After taking the first step, we will overcome the first difficulty and weaken the second difficulty, and transform a TSP into a problem of finding the minimum values, which has been thoroughly studied~\cite{chen2020maxORmin, durr1999quantumalgorithmfindingminimum, Boyer_1998}.
    When setting the appropriate number of iterations, we can obtain an optimal solution with high probability.

\subsection{TSP Encoding}
    We divide qubits into index registers and value registers, denoted separately as $\ket{index}$ and $\ket{value}$. The index registers encode possible solutions.
    For $N$-node TSP, we use $N$ sets of qubits with $m=\lceil\log N\rceil$ qubits per set as index registers. The $i$-th set of qubits encodes the node which the traveling salesman will reach from node $i$, and we always let the traveling salesman start from node $0$. Take a $4$-node TSP as example, the quantum state $\ket{2031}$ represents the tour route $0\to 2\to 3\to 1\to 0$ because the 0-th set of qubits is in state $\ket{2}$, the 2-th set of qubits is in state $\ket{3}$, etc. It is clear that a quantum state encoding a feasible solution should represent a $N$-cycle as an element of $N$-order permutation group. There are $(N-1)!$ cycles in $N$-order permutation group. Therefore, the lowest query complexity of the pure quantum search algorithm for general $N$-node TSP is $O(\sqrt{(N-1)!})$. (Of course, one may apply certain special weight structures such as Euclidean distance to improve this complexity, but in this work we research general situation.) Under the current encoding methods, if we start from the uniform superposition state on the entire encoding space directly, the query complexity will be $O(\sqrt{2^{N\lceil\log N\rceil}})=O(\sqrt{N^N})$, which makes it important to compress the search space through initial state preparation.
    The value registers encode information about the threshold and total weights of the possible tour routes. We adopt the form of the complement code, which means there is a sign bit in the value registers.
    Restricted by computing power, we set the number of qubits in the value registers to 5 for 4-7 nodes TSP and 6 for 8-node TSP. (Because of the greater weight caused by the longer length of HC for eight-node TSP, we are forced to add an additional qubit.) For convenience in simulation, we limit the weight of each edge to a positive integer. We will explain later that this restriction does not affect generality.

    Under the above encoding way, index registers need $Nm=N\lceil\log N\rceil$ qubits and value registers need $M=\lceil\log C\rceil+1$ qubits where $C$ is the maximum total weight of the possible tour route and adding one is due to the sign bit. (The edges weights can also be rescaled according to specific conditions.) Our algorithm for the TSP occupy $N\lceil\log N\rceil+\lceil\log C\rceil+1$ qubits in total, no extra ancillary qubits!

\subsection{Framework}

    Next, we sketch the components of our quantum algorithm for the TSP. It consists of three major parts (see also {\bf Fig.\ref{fig:top_design1}} and {\bf Fig.\ref{fig:top_design2}} for quantum circuits):
    \begin{enumerate}
        \item Encoding operators. All registers $\ket{index}\ket{value}$, as a group of qubits, are initialized to the uniform superposition state of all HCs with corresponding total weights by these operators, which are showed in {\bf Fig.\ref{fig:top_design1}}. $HCg$-gate can generate the uniform superposition state of all HCs as $\ket{index}$. $H^{\otimes M},\ U_{w-C_{T}}$ and inverse Quantum Fourier Transform $QFT^{\dag}$ can generate the state $\ket{value}$ according to $\ket{index}$, where $value=total\ weight\ of\ a\ HC-threshold\ C_{T}$. Hence, the output state of the circuit in {\bf Fig.\ref{fig:top_design1}} is 
        \begin{equation*}
            \frac{1}{\sqrt{(N-1)!}}\sum_{\sigma\in \{HC\}}\ket{\sigma}\ket{w_{\sigma}-C_{T}}.
        \end{equation*}
        
        \item Oracle operators. Label the states whose total weights are less than threshold $C_{T}$. In {\bf Fig.\ref{fig:top_design2}}, the $Z$-gate acting on the last qubit (i.e., sign bit) can mark negative values, that is, label the HCs whose total weights are below the threshold $C_{T}$.
        
        \item Diffusion operators. See the part behind the $Z$-gate in {\bf Fig.\ref{fig:top_design2}}; $QFT,\ U^{\dag}_{w-c_{T}}$ and $H^{\otimes M}$ can release the value registers to assist the $HCg^{\dag}$-gate, thereby our whole algorithm does not need extra auxiliary registers. 

        The part inside the dashed box in {\bf Fig.\ref{fig:top_design2}} is standard Grover's diffusion operator $HCg(I-2|0\rangle\langle0|)HCg^{\dag}$ where $H^{\otimes mN}$ is replaced with $HCg$-gate, after which we recover the value registers by $H^{\otimes M},\ U_{w-c_{T}}$ and $QFT^{\dag}$ in order to continue executing subsequent Grover's iterations.
    \end{enumerate}

    \begin{figure*}
        \centering
        \includegraphics[width=.5\linewidth]{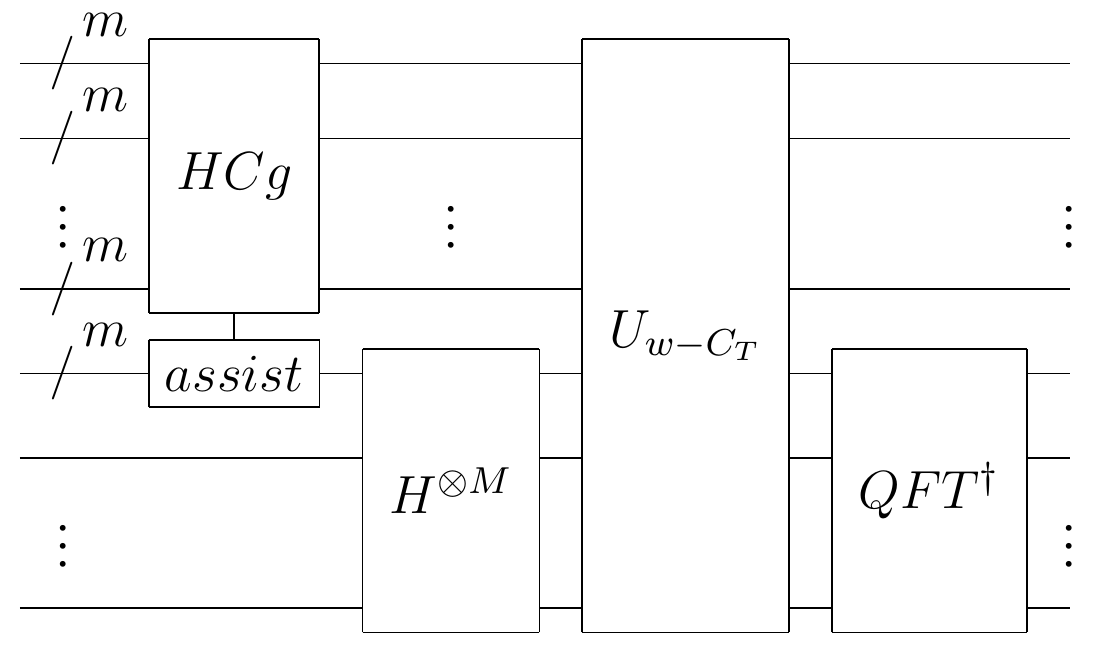}

        \caption{Illustration of the quantum circuit of our algorithm's encoding part. The input state is $\ket{0}^{\otimes (mN+M)}$, $HCg$-gate represents the HC-generation algorithm with ``assist" denoting its part on the auxiliary qubits, and $m=\lceil\log N\rceil,\ M=\lceil\log C\rceil+1$, where $C$ is maximum total weight of the possible tour route corresponding to an HC. The principles of the operators $HCg$-$gate$ and $U_{w-C_{T}}$ in the figure will be explained later.}
        \label{fig:top_design1}
    \end{figure*}

    \begin{figure*}
        \centering
        \includegraphics[width=1.\linewidth]{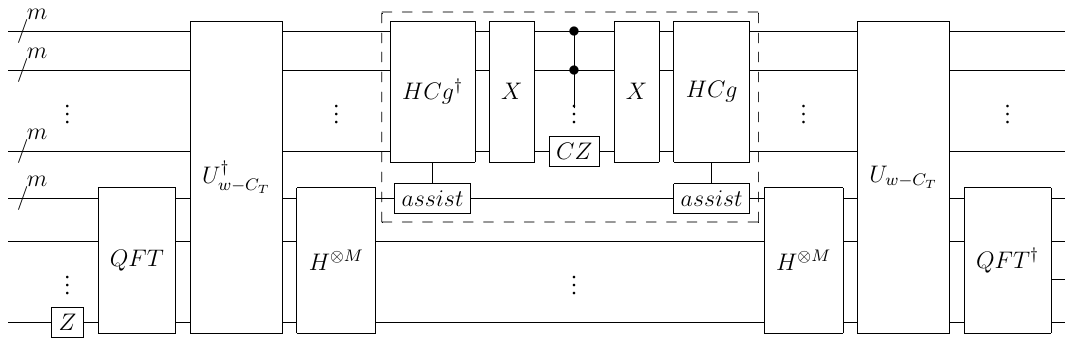}

        \caption{Illustration of the quantum circuit of our algorithm's searching module, which will be executed many times to amplify the amplitude of the optimal solution.}
        \label{fig:top_design2}
    \end{figure*}

\subsection{HC-generation algorithm}
\begin{figure*}
        \centering
        \includegraphics[width=.5\linewidth]{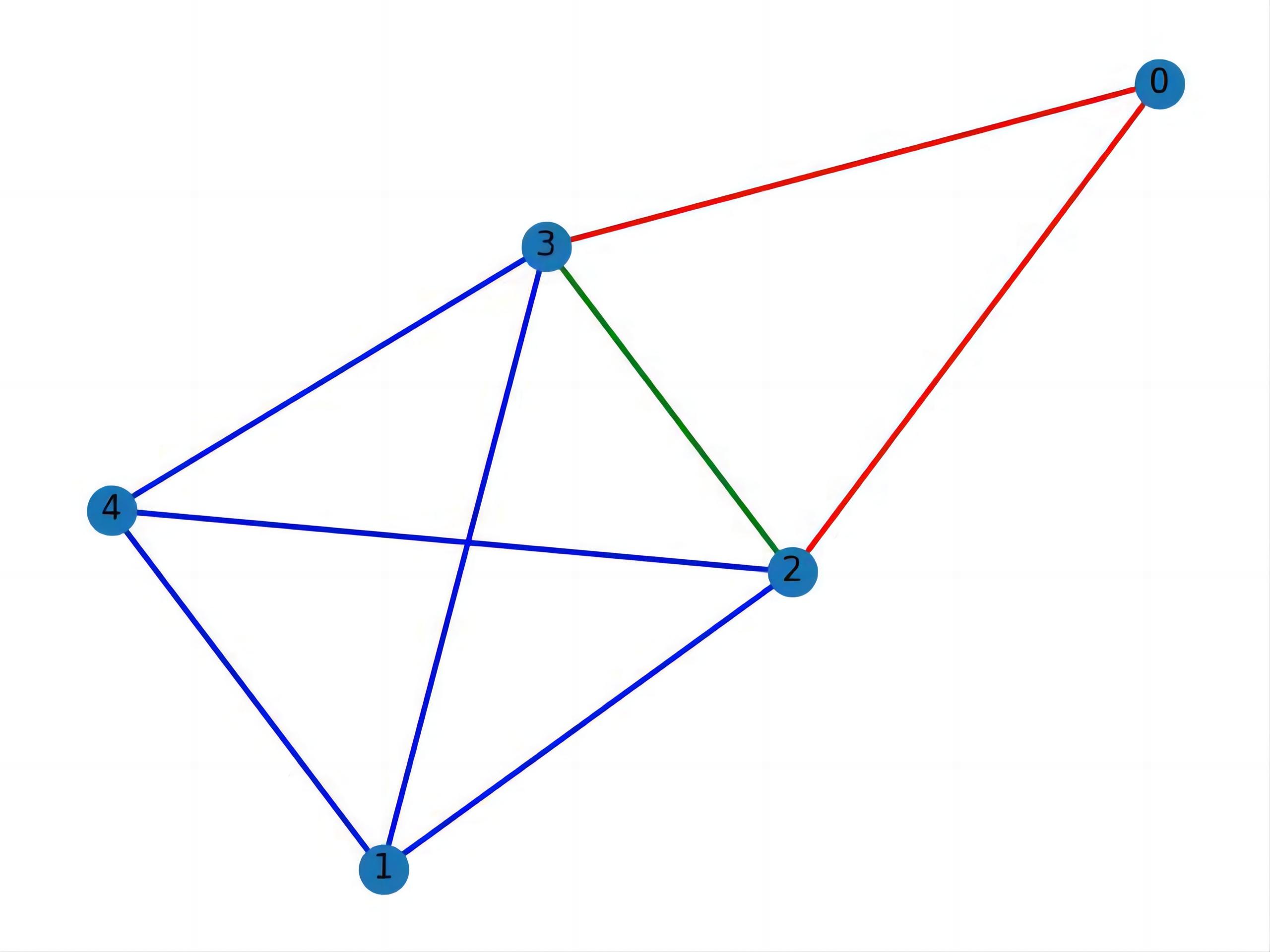}

        \caption{Obtain an HC of length 5 from an HC of length 4.}
        \label{fig:graph_example}
    \end{figure*}
    
    Firstly, we explain the motivation for designing this algorithm. When researchers used the GAS framework to solve combinatorial optimization problems in the past, they might have hoped to start from the uniform superposition state generated by the Hadamard-gate to encode the solution space. However, feasible solutions are always scattered in the encoding space, which drives researchers to design discriminant oracles to identify the valid solutions~\cite{zhu2022realizablegasbasedquantumalgorithm}. Thus, it makes sense to generate the uniform superposition state of all feasible solutions directly. It is easy to prove that preparing the uniform superposition state of all feasible solutions can limit Grover's quantum search to a subspace stretched by orthogonal basis vectors corresponding to feasible solutions. Taking the TSP as example, we can think about preparing a uniform superposition state of all permutations of order $N$. One way to achieve this goal is to construct an $indexing\ function$ that efficiently identifies each combinatorial object (i.e. permutation of order $N$) with a unique integer index. Using this function, we can unitarily map the $N!$ binary strings corresponding to permutations, which are scattered in a larger space of $2^{N\lceil\log N\rceil}$ binary strings according to our TSP encoding, to a canonical subspace of the first $N!$ binary strings (i.e., 0 through to $N!-1$). This is done using the so-called ``indexing unitary" $U_{\#}$. Once constructed $U_{\#}$, we can first prepare a uniform superposition state of 0 through to $N!-1$ by applying a Hadamard-gate and then obtain the uniform superposition state of all permutations of order $N$ by applying $U^{\dag}_{\#}$. In fact, someone has already accomplished it \cite{MYRVOLD2001281, Marsh_2020, chiew2018graphcomparisonnonlinearquantum}.

    In \cite{chiew2018graphcomparisonnonlinearquantum}, the authors consume many auxiliary registers to prepare a uniform superposition state of all permutations of order $N$. Their method essentially utilizes an $indexing$ $function$. Using the $indexing$ $function$ to prepare the superposition state is somewhat similar to encryption and decryption calculations, which require a considerable amount of computing resources such as many ancillary qubits. However, we will show that it is easier to prepare a uniform superposition state for all $N$-length HCs. Especially, we do not require an $indexing\ function$, thus avoiding the excessive occupation of auxiliary qubits.

    In classical computing, dynamic programming is based on the overlapping substructure and optimal substructure properties. We first consider HC as the optimal property we want to maintain, and then provide the classical recursive algorithm for enumerating $N$-length HCs. Finally, we utilize the property of the quantum superposition state to serve as the overlapping substructure, which achieves parallel quantum acceleration, and implement this quantum recursive process using quantum circuits. Specifically, we provide two theorems to ensure the implementation of the $HCg$-gate.

    \hspace*{\fill}
    
    \begin{theorem}[HC Recursively Enumeration Theorem]  \label{thm_HC_enumeration}
        There exists a recursive algorithm that starts from 2-length HC to enumerate all $N$-length HCs.
    \end{theorem}
    
    \begin{proof}
        Suppose that we have an HC of length $N-1$, denoted by $C$. Then, we add a new vertex $\omega$ and remove any edge of $C$, denoted by $u\to v$. By adding two new edges $u\to\omega$ and $\omega\to v$ we obtain an HC of length $N$, denoted by $\widetilde{C}$. There are $(N-2)!$ HCs of length $N-1$, each of which can generate $N-1$ different HCs by removing different edges, so they can generate $(N-1)!$ different HCs totally (because any two different HCs of the same length must have at least two different edges).
        
        Let us take the 4-length generating 5-length case as an example (see {\bf Fig.~\ref{fig:graph_example}}). We have a graph with four vertexes 1-4. For HC $1\to4\to3\to2\to1$ corresponding to permutation $\sigma=(4123)$, we can remove the green edge and add red edges maintaining orientation to obtain HC $0\to2\to1\to4\to3\to0$ corresponding to permutation $\sigma=(24103)$.

        Now, we can start from 2-length HC $(N-2)\to (N-1)\to (N-2)$ corresponding to permutation $\sigma=(N-1,N-2)$ to generate all $N$-length HCs, where in order to facilitate the achievement of quantum circuits later the last vertex added is 0.
    \end{proof}

    \hspace*{\fill}

    The above enumeration process has a terrible time complexity of $O((N-1)!)$, which makes it a NP-hard problem to find all valid HCs~\cite{akiyama1980np}. However, based on the proof of {\bf Theorem~\ref{thm_HC_enumeration}}, we can provide a quantum algorithm with polynomial gate complexity as follows.

    \hspace*{\fill}
    
    \begin{theorem}[Quantum HC-generation Theorem]  \label{thm_QDP_HC}
        There exists a quantum version of the above enumeration process that can generate the uniform superposition state of all HCs within polynomial gate complexity.
    \end{theorem}

    \begin{proof}
        We first show the enumeration process in the proof of {\bf Theorem~\ref{thm_HC_enumeration}} using our TSP encoding. Suppose we have an HC of length $N-k$, denoted by $\ket{\sigma}$ corresponding to a permutation $\sigma$ acting on the set $V=\{k, k+1, \dots, N-1\}$. We perform the following steps:
        \begin{enumerate}
            \item Add en element $k-1$ to $V$, which produces a new set $\widetilde{V}=\{k-1, k, k+1, \dots, N-1\}$ and a new permutation $\widetilde{\sigma}=(k-1,\sigma)$.
            \item Swap $k-1$ and each digit in $\sigma$ in order from small to large, which generates $N-k$ new HCs of length $N-k+1$.
        \end{enumerate}
        
        By performing the above steps for each HC of length $N-k$, we obtain $(N-k-1)!\times (N-k)=(N-k)!$ new HCs, that is, all $(N-k+1)$-length HCs. We can easily confirm that these steps are equivalent to the enumeration process in the proof of {\bf Theorem~\ref{thm_HC_enumeration}}.
        
        \begin{table}[h]
        \centering
            \begin{threeparttable}[b]
            \caption{A simple example of HC-generation.}
            \label{tab:QDP_HCg_example}
                \begin{tabular}{c|c|c|c}
                    \toprule[1.5pt]
                        Order & $N=3$ & $N=4$ & $N=5$ \\ \hline 
                        Identity permutation & 234 & 1234 & 01234 \\ \hline 
                        \multirow{9}{*}{HC} & $\ket{342}$ & \textcolor{red}{$\ket{2341}$} & \textcolor{red}{$\ket{12340}$} \\
                        & $\ket{423}$ & \textcolor{red}{$\ket{3142}$} & \textcolor{red}{$\ket{20341}$} \\
                        & - & \textcolor{red}{$\ket{4312}$} & \textcolor{red}{$\ket{32041}$} \\
                        & - & $\ket{2413}$ & \textcolor{red}{$\ket{42301}$} \\
                        & - & $\ket{3421}$ & $\ket{13042}$ \\
                        & - & $\ket{4123}$ & $\ket{23140}$ \\
                        & - & - & $\ket{30142}$ \\
                        & - & - & $\ket{43102}$ \\
                        & - & - & $\cdots$ \\
                    \bottomrule[1.5pt]
                \end{tabular}
            \end{threeparttable}
        \end{table}

        For example, we start from 3-length HCs to generate 4-length and 5-length HCs in {\bf Table~\ref{tab:QDP_HCg_example}}. When $N=3$, we perform the above two steps for $\ket{342}$ to generate the states marked in red in the column with $N=4$. When $N=4$, we perform the above two steps for $\ket{2341}$ to generate four states marked in red in the column with $N=5$.

        In the following section, we provide a quantum implementation of the above process. It should be noted that although we mentioned ``swap" in step two, we do not actually need $Swap$-gates. Instead, we first prepare a uniform superposition state of integers not less than $k$. Consider an example of generating HCs of $N=5$ from HCs of $N=3$, where $k=5-3=2$. We should first generate HCs of $N=4$. Suppose we have the state (Refer to {\bf Table~\ref{tab:QDP_HCg_example}})
        \begin{equation*}
            \frac{1}{\sqrt{2}}\ket{0}(\ket{342}+\ket{423}),
        \end{equation*}
        then we prepare a uniform superposition state from $\ket{0}$ by Amplitude Amplification Method \cite{Brassard_2002} to get the state:
        \begin{align*}
            {} & \frac{1}{\sqrt{2}} \cdot \frac{1}{\sqrt{3}} (\ket{2}+\ket{3}+\ket{4})\cdot(\ket{342}+\ket{423}) \\
             = {} & \frac{1}{\sqrt{6}}(\ket{2}\ket{342}+\ket{3}\ket{342}+\ket{4}\ket{342}+\dots)
        \end{align*}

        Now, we need to match the string of numbers in the second register with the number in the first register, and replace the matched number in the string with 1. (In practice, for convenience we will always first replace it with 0 then replace 0 with the original number.) Take the second item on the right-hand side of the equation above as an example:
        \begin{equation}
            \ket{3}\ket{342} \rightarrow \ket{3}\ket{042} \rightarrow \ket{3142},
            \label{eqn:HCg_example}
        \end{equation}
        and corresponding quantum circuit is in {\bf Fig.\ref{fig:HC-generation}}, in which the first set of registers is ancillary register and $\ket{\psi}=\frac{1}{\sqrt{3}} (\ket{2}+\ket{3}+\ket{4})$. Each number is encoded in binary, therefore, each set of registers requires $\lceil\log 5\rceil=3$ qubits. In the figure, we omit the registers where the other numbers in $\ket{342}$ are located except for 3. The operations performed on these omitted registers are the same as those shown in {\bf Fig.\ref{fig:HC-generation}} except for module D.

        In the following, we explain the principles of the four modules in {\bf Fig.\ref{fig:HC-generation}} using (\ref{eqn:HCg_example}) as an example.

        \begin{figure*}
            \centering
            \includegraphics[width=.9\linewidth]{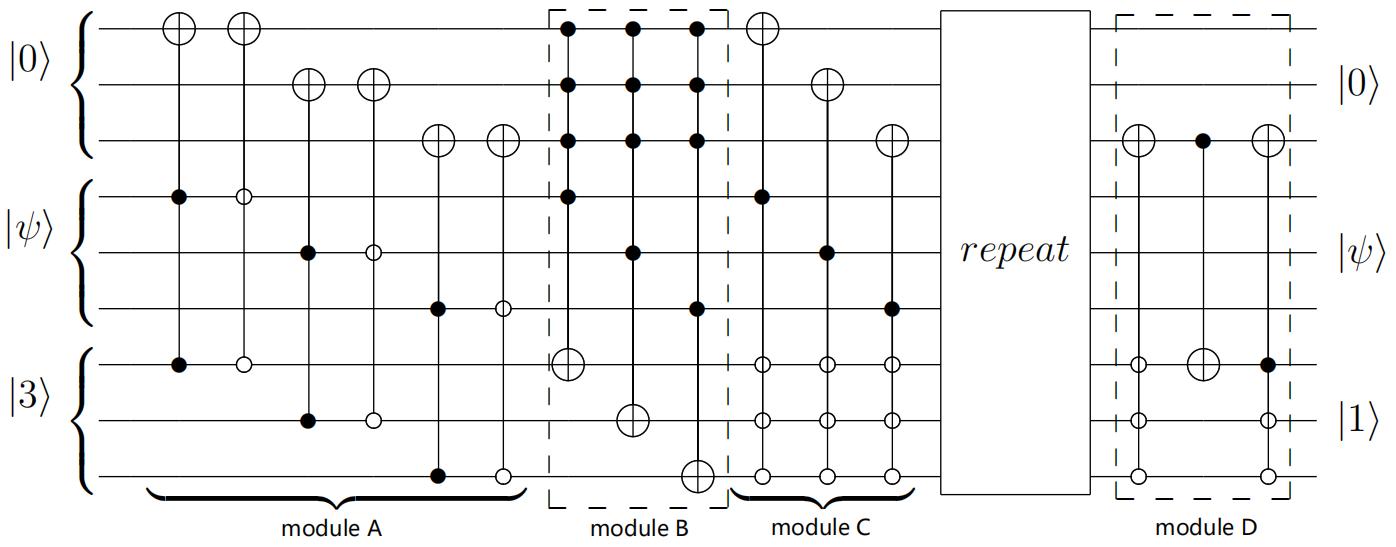}
    
            \caption{Partial quantum circuit of HC-generation algorithm, where the first set of qubits is auxiliary register initialized to $|0^{\log N}\rangle$. To ensure that there are always available qubits as ancillary register for each iteration, this algorithm requires additional $\log N$ auxiliary qubits except for the index registers.}
            \label{fig:HC-generation}
        \end{figure*}

        \begin{enumerate}
            \item Module A. The function is to match the second and third sets of registers. If the numbers in the second and third sets of registers are exactly the same, three qubits of the first set of registers will become $\ket{111}$.
            \item Module B. The function is to set the third set of registers to $\ket{0}$ by performing $XOR$ operation on the second and third sets of registers when the first set of registers is in state $\ket{111}$.
            \item Module C. The function is to free up the first set of registers. This module must be combined with the $repeat$-gate behind it. The $repeat$-gate is a simple repetition of module A. If the third set of registers is not in state $\ket{000}$, we only need to execute the $repeat$-gate to free up the first set of registers. However, if the third set of registers is in state $\ket{000}$, we have to discuss different situations separately: 
            
            Consider the first qubit's release as an example. If the fourth qubit in module C is $\ket{1}$, the operator performed on the first qubit in module A must be $CCX$-gate. But for $repeat$-gate, the same gate will not execute because the seventh qubit is $\ket{0}$. Hence, only the controlled $NOT$-gate in module C will act on the first qubit. 

            If the fourth qubit in module C is $\ket{0}$, the operator performed on the first qubit in module A must be $anti$-$C\ anti$-$C\ X$-gate. Thus, the $anti$-$C\ anti$-$C\ X$-gate in the $repeat$-gate will be executed, but the first controlled $NOT$-gate in module C will not.

            In summary, only one controlled $NOT$-gate in module C or $repeat$-gate will act on the first qubit, avoiding duplicate execution, which ensures the correct release of the first qubit. The other qubits in the ancillary register can be released in the same way.
            \item Module D. The function is to execute the second step in (\ref{eqn:HCg_example}), i.e., $\ket{042}\to \ket{142}$, which only needs an ancillary qubit. This module can be modified according to the number changing from 0, that is, modifying the target qubit of the second gate and the control qubit of the third gate in this module according to the binary sequence of the corresponding numbers.
        \end{enumerate}

        The functions of these four modules can be verified according to the corresponding quantum circuit and promote them to handle more nodes by expanding qubits. We also provide numerical simulation verification of 4-length and 5-length HCs in the supplementary materials.

        Finally, we analyze the gate complexity of the quantum algorithm described above.
        We can realize the Amplitude Amplification Method (AAM) using Grover algorithm with zero theoretical failure rate \cite{Long_2001}, where Grover's iteration is:

        \begin{equation}
            G_{k}(\phi,\varphi)=H^{\otimes m}S_0(\phi)H^{\otimes m}S_{\chi}(\varphi),\ m=\lceil\log N\rceil
            \label{eqn:AAM1}
        \end{equation}
        \begin{equation}
            S_0(\phi)\ket{j}=\left\{ \begin{aligned}  
                e^{i\phi}\ket{j},\ j=0\\
                \ket{j},\ j\neq0
            \end{aligned}\right.
            \label{eqn:AAM2}
        \end{equation}
        \begin{equation}
            S_{\chi}(\varphi)\ket{j}=\left\{ \begin{aligned}  
                e^{i\varphi}\ket{j},\ k\leq j < N\\
                \ket{j},\ 0\leq j < k
            \end{aligned}\right.
            \label{eqn:AAM3}
        \end{equation}
        
        We adopt three dimensional rotation form according to \cite{Long_2001}:
        \begin{gather}
            \ket{\psi}=G^n_k(\phi,\phi)H^{\otimes m}\ket{0}\\
            n=\lceil\frac{\pi}{4\arcsin{\sqrt{k/2^m}}}-\frac{1}{2}\rceil,\ \phi=2\arcsin{\frac{\sin{\frac{\pi}{4n+2}}}{\sqrt{k/2^m}}} \notag
            \label{eqn:AAM4}
        \end{gather}

        \begin{figure*}
            \centering
            \includegraphics[width=.7\linewidth]{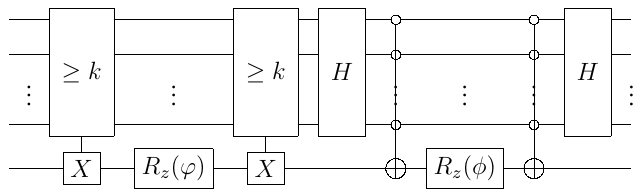}
    
            \caption{Quantum circuit of $G_{k}(\phi,\varphi)$, where the last qubit is ancillary register with $\ket{0}$ as input and output, and the ``$\geq k$" gate is an integer comparator that conditionally toggles the ancillary register if the input registers hold a value not less than $k$.}
            \label{fig:AAM}
        \end{figure*}
        
        Grover's iteration $G_{k}(\phi,\varphi)$ can be performed using the quantum circuit shown in {\bf Fig.\ref{fig:AAM}}. To realize the ``$\geq k$" gate, we can use a set of $MCX$ gates to match numbers not smaller than $k$. Hence, we need $O(N)$ $MCX$ gates, $O(\log N)$ $H$ gates and $2$ $RZ$ gates to execute $G_k(\phi,\phi)$; and $O(N\sqrt{N})$ $MCX$ gates, $O(\sqrt{N}\log N)$ $H$ gates and $2\sqrt{N}$ $RZ$ gates to execute the AAM. Next, we focus on the $MCX$ gates shown in {\bf Fig.\ref{fig:HC-generation}}, where each set of registers requires $\lceil\log N\rceil$ qubits for $N$-node situation, requiring $O(\log N)$ $MCX$ gates. Considering that there are $N$ sets of qubits ($\log N$ qubits per set) in the index registers, each recursive step has $O(N\log N+N\sqrt{N})$ $MCX$ gates (the second term is from the AAM), $O(\sqrt{N}\log N)$ $H$ gates and $2\sqrt{N}$ $RZ$ gates. The total number of recursive steps is $N-2$ (because the algorithm starts from 2-length HC to generate all $N$-length HCs). Therefore, our HC-generation algorithm requires $O(N^2\log N+N^2\sqrt{N})=O(N^{5/2})$ $MCX$ gates, $O(N^{3/2}\log N)$ $H$ gates and $O(N^{3/2})$ $RZ$ gates. The total gate complexity is $O(N^{5/2})$, which is a polynomial complexity.
        
    \end{proof}

    \hspace*{\fill}

    This theorem enables our algorithm's encoding part (see {\bf Fig.~\ref{fig:top_design1}}) to prepare a good initial state in the polynomial complexity, which achieves exponential acceleration compared to the previous initial state preparation algorithm (we will discuss it in the ``Discussion'' section).
    It should be noted that although our algorithm has polynomial complexity, this does not mean that we can achieve exponential acceleration compared to the classical algorithm because the preparation of the uniform superposition state of all HCs and enumerating all HCs are two different problems.

\subsection{Shortcut of $QFT$} 
    \cite{Gilliam_2021CPBO} provides a detailed introduction to GAS for Constrained Polynomial Binary Optimization (CPBO), which employs a shortcut of $QFT$. This idea can be extended to our problem.

    In the previous section, we realized the HC-generation algorithm as the $HCg$-gate in {\bf Fig.~\ref{fig:top_design1}} and {\bf Fig.~\ref{fig:top_design2}}. In the following, we will realize the $U_{w-C_{T}}$-gate to compute difference value between the total weight of an HC and the threshold $C_{T}$ according to the index registers $\ket{index}$.
    
    First, we can obtain the geometric sequence encoding of an integer $-2^{M-1}\leq k<2^{M-1}$ by applying the circuit in {\bf Fig.\ref{fig:GAS for CPBO}} to state $\ket{0^M}$.
    For $R$-gate can rotate the amplitudes of states having 1 in the position corresponding to the qubit it was applied to, the state obtained after the dashed box in {\bf Fig.\ref{fig:GAS for CPBO}} is $\frac{1}{\sqrt{2^M}}\sum^{2^M-1}_{j=0}e^{ij\theta}\ket{j}_M$, which represents a geometric sequence of length $2^M$.

    \begin{figure}[H]
        \centering
        \includegraphics[width=0.8\linewidth]{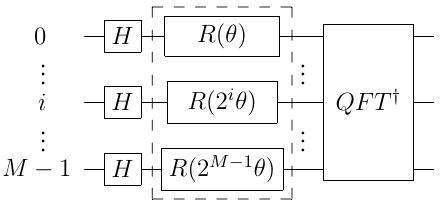}

        \caption{Quantum circuit for getting the geometric sequence encoding of an integer $-2^{M-1}\leq k<2^{M-1}$, where the part inside the dashed box is defined as $U_G(\theta),\ \theta\in [-\pi,\pi)$ and $R$-gate is phase gate.}
        \label{fig:GAS for CPBO}
    \end{figure}

    We know $QFT$ applied to a quantum state encoding a binary non-negative integer also creates a geometric sequence of amplitudes. Hence, if the encoded numbers are known classically, $U_G(\theta)$ can be regarded as a shortcut of the $QFT$ because there are no multi-qubit interactions in $U_G(\theta)$. Now we can set $\theta=2\pi k/2^M\ (-2^{M-1}\leq k<2^{M-1})$ to obtain $\ket{k\ mod\ 2^M}$ as the output state of the circuit in {\bf Fig.\ref{fig:GAS for CPBO}}.

    This representation of $k$ takes the number with the highest digit of one as a negative number, similar to the complementary codes in classical computer science.

    \cite{Gilliam_2021CPBO} use $U_G(\theta)$ controlled by boolean variables to solve CPBO problems. We can extend this method to our quantum algorithm for the TSP. Under our TSP encoding, every index register encodes the result $\sigma(i)$ produced after an HC as a permutation $\sigma$ acts on $i\in \{0,1,\dots,N-1\}$. To compute the total weight of an HC, we use these index registers as the control registers of $U_G(\theta)$:

     \begin{figure}[H]
        \centering
        \includegraphics[width=1.\linewidth]{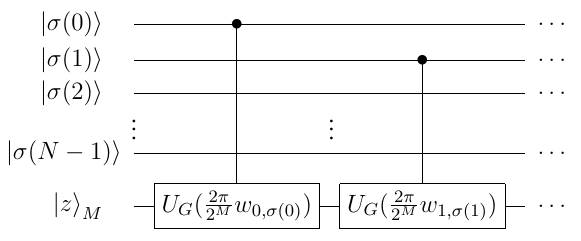}

        \caption{The controlled-$U_G(\theta)$ in our quantum algorithm for TSP, where $\ket{z}_M=H^{\otimes M}\ket{0}_M,\ w_{i,\sigma(i)}$ is the weight of the edge $(i,\sigma(i))$.}
        \label{fig:CPBO4}
    \end{figure}

   However, we do not know $\sigma(i)$ for the index registers are in the superposition state. Therefore, we need to prepare a controlled-$U_G(\theta)$ for every possible situation:

   \begin{figure}[H]
        \centering
        \includegraphics[width=1.\linewidth]{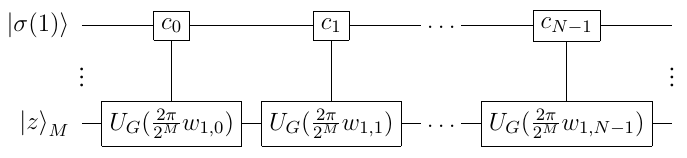}

        \caption{Refinement of the second controlled-$U_G(\theta)$ in {\bf Fig.~\ref{fig:CPBO4}}, where $c_0,c_1,\dots,c_{N-1}$ match integers $0,1,\dots,N-1$ according to binary expansion in sequence. For example, if $\sigma(1)=1$, the second gate $U_G(2\pi\omega_{1,1}/2^M)$ will execute.}
        \label{fig:CPBO5}
    \end{figure}

    Now we have realized $U_{w-C_{T}}$-gate in {\bf Fig.\ref{fig:top_design1}} and {\bf Fig.\ref{fig:top_design2}}, which requires $N^2\ U_G(\theta)$-gates. (An additional gate $U_G(-2\pi C_T/2^M)$ needs to be added at the end, then, we can obtain $\ket{w_\sigma-C_T\ mod\ 2^M}$ in the value registers where $\sigma$ is an HC.) We only need a $Z$-gate on the last qubit to judge the symbol of $w-C_T$.

    $U_G(\theta)$ can be seen as a shortcut of the $QFT$ because it avoids multi-qubit interactions. Recall that we must construct a unitary matrix that takes the weights of HCs as the phases of its eigenvalues if we use Quantum Phase Estimation (QPE) to estimate the weights of HCs. To execute this unitary matrix, we must recursively decompose it into simple basic operations~\cite{zhu2022realizablegasbasedquantumalgorithm}, which is not very easy. In contrast, our algorithm only requires simple controlled-$U_G(\theta)$, which can be directly executed using controlled phase gates. Consequently, we can complete the calculation of the value registers at a cost of $O(M^2+N^2M)$ ($M^2$ for $QFT^{\dag}$) controlled phase gates. Someone may question whether the above method can handle non-integer values. Hereto, relevant issues have been discussed in \cite{Gilliam_2021CPBO}, in which two methods are introduced to handle general real numbers, i.e., approximating real coefficients by fractions and encoding real coefficients as Fejér distributions.

    By now, one can understand the entire algorithm by combining {\bf Fig.\ref{fig:top_design1}}, {\bf Fig.\ref{fig:top_design2}} and the previous section ``framework".

\section{Results}
In this section, we implement our proposed algorithm for graphs with node numbers $N=4, 5, 6, 7, 8$ and provide the numerical results. We only consider complete graphs for one can always transform incomplete graphs into complete graphs by adding edges with sufficiently large weights. This method can also help to transform a Hamiltonian circle problem into a TSP. In the following, we present the results on seven examples. We take $M=5$ for 4-7 nodes TSP and $M=6$ for eight-node TSP. Specifically, the adjacency matrices of these seven graphs are:
    \begin{equation}
        \mathbf{X_1}=\left(
        \begin{array}{cccc}
            0 & 1 & 1 & 3 \\
            1 & 0 & 2 & 1 \\
            1 & 2 & 0 & 1 \\
            3 & 1 & 1 & 0 \\
        \end{array}\right),\
        \mathbf{X_2}=\left(
        \begin{array}{cccc}
            0 & 2 & 1 & 3 \\
            2 & 0 & 2 & 1 \\
            1 & 2 & 0 & 3 \\
            3 & 1 & 3 & 0 \\
        \end{array}\right)\
        \label{eqn:node4 adjacency matrices}
    \end{equation}

    \begin{equation}
        \mathbf{X_3}=\left(
        \begin{array}{ccccc}
            0 & 1 & 3 & 2 & 2 \\
            1 & 0 & 1 & 3 & 1 \\
            3 & 1 & 0 & 2 & 2 \\
            2 & 3 & 2 & 0 & 1 \\
            2 & 1 & 2 & 1 & 0 \\
        \end{array}\right),\
        \mathbf{X_4}=\left(
        \begin{array}{ccccc}
            0 & 1 & 3 & 2 & 1 \\
            1 & 0 & 1 & 3 & 1 \\
            3 & 1 & 0 & 2 & 3 \\
            2 & 3 & 2 & 0 & 1 \\
            1 & 1 & 3 & 1 & 0 \\
        \end{array}\right)\
        \label{eqn:node5 adjacency matrices}
    \end{equation}

    \begin{equation}
        \mathbf{X_5}=\left(
        \begin{array}{cccccc}
            0 & 1 & 3 & 2 & 2 & 1 \\
            1 & 0 & 1 & 3 & 1 & 2 \\
            3 & 1 & 0 & 3 & 2 & 3 \\
            2 & 3 & 3 & 0 & 1 & 1 \\
            2 & 1 & 2 & 1 & 0 & 3 \\
            1 & 2 & 3 & 1 & 3 & 0 \\
        \end{array}\right)\
        \label{eqn:node6 adjacency matrices}
    \end{equation}

    \begin{equation}
        \mathbf{X_6}=\left(
        \begin{array}{ccccccc}
            0 & 1 & 3 & 2 & 2 & 1 & 1 \\
            1 & 0 & 1 & 3 & 1 & 2 & 1 \\
            3 & 1 & 0 & 3 & 2 & 3 & 1 \\
            2 & 3 & 3 & 0 & 1 & 1 & 1 \\
            2 & 1 & 2 & 1 & 0 & 3 & 1 \\
            1 & 2 & 3 & 1 & 3 & 0 & 1 \\
            1 & 1 & 1 & 1 & 1 & 1 & 0 \\
        \end{array}\right)\
        \label{eqn:node7 adjacency matrices}
    \end{equation}

    \begin{equation}
        \mathbf{X_7}=\left(
        \begin{array}{cccccccc}
            0 & 1 & 3 & 2 & 1 & 1 & 2 & 1 \\
            1 & 0 & 1 & 3 & 1 & 2 & 1 & 2 \\
            3 & 1 & 0 & 1 & 2 & 2 & 1 & 3 \\
            2 & 3 & 1 & 0 & 1 & 2 & 1 & 1 \\
            1 & 1 & 2 & 1 & 0 & 3 & 1 & 3 \\
            1 & 2 & 2 & 2 & 3 & 0 & 1 & 2 \\
            2 & 1 & 1 & 1 & 1 & 1 & 0 & 1 \\
            1 & 2 & 3 & 1 & 3 & 2 & 1 & 0 \\
        \end{array}\right)\
        \label{eqn:node8 adjacency matrices}
    \end{equation}

    The same quantum circuit was executed for $shots=1000$ times to obtain statistical data of final quantum state $\ket{index}$, where although we measured both sets of registers $\ket{index}\ket{value}$ simultaneously in the end, we only extracted the measurement results of the first set. To confirm the effectiveness of the algorithm, we set the threshold $C_T=w_{\sigma^*}+1$ where $w_{\sigma^*}$ is the total weight of the optimal solution. Of course, in practice the optimal solution is unknown at the beginning, therefore, the threshold should be updated from an initial assumption value by executing the algorithm multiple times.
    \hspace*{\fill}

    {\bf Table~\ref{tab:sim}} shows our simulation results on IBM's qiskit, where ``qubits" denotes the required number of qubits, ``opt num" denotes the number of the optimal solution, ``iterations" denotes the optimal number of iterations, $x=\lceil\sqrt{(N-1)!/opt\ num}\rceil$ and ``accuracy" is computed from the proportion of the optimal solutions in 1000 samples. We also provide the bar charts in the supplementary materials, including two examples of our HC-generation algorithm.
    
    \begin{table}[H]
        \centering
        \begin{threeparttable}[b]
            \caption{Simulation Results} \label{tab:sim}
            \setlength{\tabcolsep}{2mm}
            \begin{tabular}{c c c c c c}
                \toprule[1.5pt]
                \     &  $N$ & qubits & opt num & iterations    & accuracy    \\ \hline
                $X_1$ &  4   & 13     & 2       & $5x+1$        & $100\%$     \\
                $X_2$ &  4   & 13     & 4       & $x$           & $99.8\%$    \\
                $X_3$ &  5   & 20     & 4       & $3x$          & $98.7\%$    \\
                $X_4$ &  5   & 20     & 2       & $3x+1$        & $99.9\%$    \\
                $X_5$ &  6   & 23     & 2       & $5x+2$        & $100\%$     \\
                $X_6$ &  7   & 26     & 4       & $5x+3$        & $99.9\%$    \\
                $X_7$ &  8   & 30     & 6       & $5x+13$        & $99.7\%$    \\
                \bottomrule[1.5pt]
            \end{tabular}
        \end{threeparttable}
    \end{table}

\section{Discussion}
We have transformed the TSP into a minimum value search problem, in which there are two problems. One is how to provide a good Grover's iteration count setting scheme when the number of the optimal solutions is unknown. The other is how to search for the minimum value by executing our algorithm multiple times to update the threshold $C_T$. These two problems can be solved by applying quantum exponential searching~\cite{durr1999quantumalgorithmfindingminimum, Boyer_1998}. After these improvements, the final algorithm still maintains the query complexity of $O(\sqrt{(N-1)!})$ (simply adding a constant factor). The complete algorithm steps, which are the general steps of Grover's algorithm for finding the minimum value (just change the number of feasible solutions), are as follows:

    \begin{enumerate}
        \item Randomly select an HC and take its total weight as the initial threshold. Initialize $l=1$ and set $\lambda=6/5$. (Any value of $\lambda$ strictly between 1 and $4/3$ will perform. Parameters $l$ and $\lambda$ are used to determine the number of iteration steps.)
        \item Repeat the following and interrupt it when the total number of executions of Grover's searching module is more than $22.5\sqrt{(N-1)!}$. (Grover's searching module can be seen in {\bf Fig.\ref{fig:top_design2}}. The constant coefficient 22.5 depends on $\lambda$.)
        \begin{enumerate}[a]
            \item Execute the encoding module shown in {\bf Fig.\ref{fig:top_design1}} to prepare the initial state.
            \item Apply the quantum exponential searching algorithm where Grover's iteration can be seen in {\bf Fig.\ref{fig:top_design2}}:
            \begin{enumerate}[i]
                \item Choose $j$ uniformly at random among the non-negative integers smaller than $l$.
                \item Apply $j$ iterations of Grover's searching module.
            \end{enumerate}
            \item Observe the measurement result. If the total weight $w_\sigma$ of the HC obtained in the output is less than the threshold, $w_\sigma$ is used to update the threshold. Set $l=min(\lambda l,\sqrt{(N-1)!})$ and return to step (a).
        \end{enumerate}
        
        \item Return the final measurement result.
    \end{enumerate}

    \hspace*{\fill}
    
    The probability finding the optimal solution for the above quantum algorithm is at least $1/2$ by computing the expected running time to find the minimum value, given by \cite{durr1999quantumalgorithmfindingminimum, Boyer_1998}. Therefore, we can run $c$ times to ensure a success probability of at least $1-(\frac{1}{2})^c$.

    {\bf Qubit consumption.}
        Reducing qubit consumption is of practical importance, both from a simulating point of view, and from a NISQ-implementable perspective. 
        In our proposed algorithm, the consumption of qubits has been reduced to the extreme because we only need index and value registers which are necessary to store the information required to solve the TSP. Note that if ancillary qubits in our HC-generation algorithm are more than $M$, we will need extra qubits. But if the weight of a certain edge is larger than $N$, we will need at least $\lceil\log N^2\rceil$ qubits for value registers to store the possible maximum total weight of HC, which is larger than $\lceil\log N\rceil$ required by HC-generation algorithm. So without loss of generality, we can always assume $M>\lceil\log N\rceil$. Finally, our algorithm needs $N\lceil\log N\rceil+M$ qubits in total.

    {\bf Gate complexity.}
        Gate complexity is closely related to the depth of quantum circuit and the running time of the algorithm, which is an important indicator. The total gate complexity of our HC-generation algorithm is $O(N^{5/2})$. Grover's diffusion operator inside the dashed box in {\bf Fig.\ref{fig:top_design2}} requires one $HCg$-gate and its inverse, $2N\log N$ $X$ gates, one $MCZ$ gate, therefore, its total gate complexity is $O(N^{5/2})$. $U_{\omega-C_{T}}$-gate in {\bf Fig.\ref{fig:top_design1}} and {\bf Fig.\ref{fig:top_design2}} can be executed by $O(N^2M)$ multi-qubit controlled phase gates. The gate complexity of $QFT$ is $O(M^2)$. Therefore, the gate complexity of our algorithm's encoding part and searching module in {\bf Fig.\ref{fig:top_design1}} and {\bf Fig.\ref{fig:top_design2}} are both $O(N^{5/2}+M+N^2M+M^2)=O(N^{5/2}+N^2M+M^2)$ ($M$ comes from $H^{\otimes M}$, $M>\lceil\log N\rceil$), and the total gate complexity of our whole algorithm is $O((N^{5/2}+N^2M+M^2)\sqrt{(N-1)!})$.

    {\bf Query complexity.}
        Owing to the varying complexity using different decomposition methods of gates and encoding methods of data, when comparing the complexity of algorithms we generally use query complexity, that is, the number of execution times of oracle marking the target states. In {\bf Table~\ref{tab:compare}}, we list the comparison results with excellent realizable quantum search algorithms for the TSP. We compare our algorithm with GAS(HCD) in the average sense. And $O(\sqrt{(N-1)!})$ is the theoretical minimum query complexity of quantum search algorithms for a general TSP because we have to find the shortest one among $(N-1)!$ HCs when we consider random unknown input instances. One may note that for sparse instances with very small $d$, GAS(HCD) has lower query complexity, while what we care about are instances on the general complete graphs, which contain situations of sparse graphs by setting weights on certain edges to 0.

        \hspace*{\fill}

        \begin{table}[H]
            \centering
            \begin{threeparttable}[b]
                \caption{Comparison Results on query complexity. $d$ is the maximum degree of the graph and $d\sim O(N)$ under the statistical significance of random instances.} \label{tab:compare}
                
                \setlength{\tabcolsep}{2mm}
                \begin{tabular}{c c}
                    \toprule[1.5pt]
                    algorithm    &  query complexity   \\ \hline
                    GAS(HCD)~\cite{zhu2022realizablegasbasedquantumalgorithm}  & $O(\sqrt{2^{\lceil\log d\rceil N}})\approx O(\sqrt{N^N})$     \\
                    TSQS(HOBO)~\cite{sato2024circuitdesigntwostepquantum}  &  $O(\sqrt{N!})$ \\
                    Our algorithm  & $O(\sqrt{(N-1)!})\approx O(\sqrt{N^N/e^N})$  \\
                    \bottomrule[1.5pt]
                \end{tabular}
                
            \end{threeparttable}
        \end{table}
        
        The query complexity of the first step of the TSQS(HOBO) algorithm (preparing the initial state) is $O(\sqrt{2^{N\log N}/N!})\approx O(e^{N/2}/N^{1/4})$ according to Sterling's approximation~\cite{sato2024circuitdesigntwostepquantum}. Our algorithm's step of preparing the initial state includes generating the uniform superposition state of all HCs and calculating the corresponding weights (see {\bf Fig.\ref{fig:top_design1}}), not using search strategy, therefore, it is not suitable to discuss query complexity. However, the gate complexity of the encoding part in our algorithm is $O(N^{5/2}+N^2M+M^2)$, which means that we can prepare the initial state within the polynomial gate complexity, achieving exponential acceleration. ($M$ is the number of qubits in the value registers, which can be determined based on the precision of the calculation and the range of weights that need to be stored.)

    Overall, our algorithm for the TSP performed well in terms of query complexity and accuracy. In addition, our algorithm has no black boxes and all quantum gates are simple, which makes it easier to implement on real quantum hardware.

\section*{Acknowledgments}
    This work was supported by the National Key R\&D Program of China (Grant No. 2023YFA1009403), the National Natural Science Foundation special project of China (Grant No.12341103) and National Natural Science Foundation of China (Grant No. 62372444).

\bibliographystyle{naturemag}
\bibliography{nature_template}

\begin{thebibliography}{10}
\expandafter\ifx\csname url\endcsname\relax
  \def\url#1{\texttt{#1}}\fi
\expandafter\ifx\csname urlprefix\endcsname\relax\def\urlprefix{URL }\fi
\providecommand{\bibinfo}[2]{#2}
\providecommand{\eprint}[2][]{\url{#2}}

\bibitem{montanaro2016quantum}
\bibinfo{author}{Montanaro, A.}
\newblock \bibinfo{title}{Quantum algorithms: an overview}.
\newblock \emph{\bibinfo{journal}{npj Quantum Information}} \textbf{\bibinfo{volume}{2}}, \bibinfo{pages}{1--8} (\bibinfo{year}{2016}).

\bibitem{Frank2019Quantum}
\bibinfo{author}{Arute, F.} \emph{et~al.}
\newblock \bibinfo{title}{Quantum supremacy using a programmable superconducting processor}.
\newblock \emph{\bibinfo{journal}{Nature}} \textbf{\bibinfo{volume}{574}}, \bibinfo{pages}{505–510} (\bibinfo{year}{2019}).
\newblock \urlprefix\url{https://www.nature.com/articles/s41586-019-1666-5}.

\bibitem{shor1994algorithms}
\bibinfo{author}{Shor, P.~W.}
\newblock \bibinfo{title}{Algorithms for quantum computation: {D}iscrete logarithms and factoring}.
\newblock In \emph{\bibinfo{booktitle}{Proceedings 35th annual symposium on foundations of computer science}}, \bibinfo{pages}{124--134} (\bibinfo{organization}{IEEE}, \bibinfo{year}{1994}).

\bibitem{grover1996fast}
\bibinfo{author}{Grover, L.~K.}
\newblock \bibinfo{title}{A fast quantum mechanical algorithm for database search}.
\newblock In \emph{\bibinfo{booktitle}{Proceedings of the twenty-eighth annual ACM symposium on Theory of computing}}, \bibinfo{pages}{212--219} (\bibinfo{year}{1996}).

\bibitem{boyer1998tight}
\bibinfo{author}{Boyer, M.}, \bibinfo{author}{Brassard, G.}, \bibinfo{author}{H{\o}yer, P.} \& \bibinfo{author}{Tapp, A.}
\newblock \bibinfo{title}{Tight bounds on quantum searching}.
\newblock \emph{\bibinfo{journal}{Fortschritte der Physik: Progress of Physics}} \textbf{\bibinfo{volume}{46}}, \bibinfo{pages}{493--505} (\bibinfo{year}{1998}).

\bibitem{PhysRevLett.103.150502}
\bibinfo{author}{Harrow, A.~W.}, \bibinfo{author}{Hassidim, A.} \& \bibinfo{author}{Lloyd, S.}
\newblock \bibinfo{title}{Quantum algorithm for linear systems of equations}.
\newblock \emph{\bibinfo{journal}{Phys. Rev. Lett.}} \textbf{\bibinfo{volume}{103}}, \bibinfo{pages}{150502} (\bibinfo{year}{2009}).
\newblock \urlprefix\url{https://link.aps.org/doi/10.1103/PhysRevLett.103.150502}.

\bibitem{2000quant.ph..1106F}
\bibinfo{author}{{Farhi}, E.}, \bibinfo{author}{{Goldstone}, J.}, \bibinfo{author}{{Gutmann}, S.} \& \bibinfo{author}{{Sipser}, M.}
\newblock \bibinfo{title}{{Quantum Computation by Adiabatic Evolution}}.
\newblock \emph{\bibinfo{journal}{arXiv e-prints}} \bibinfo{pages}{quant--ph/0001106} (\bibinfo{year}{2000}).
\newblock \eprint{quant-ph/0001106}.

\bibitem{nielsen2002quantum}
\bibinfo{author}{Nielsen, M.~A.} \& \bibinfo{author}{Chuang, I.}
\newblock \bibinfo{title}{Quantum computation and quantum information} (\bibinfo{year}{2002}).

\bibitem{Farhi2014AQA}
\bibinfo{author}{Farhi, E.}, \bibinfo{author}{Goldstone, J.} \& \bibinfo{author}{Gutmann, S.}
\newblock \bibinfo{title}{A quantum approximate optimization algorithm}.
\newblock \emph{\bibinfo{journal}{arXiv: Quantum Physics}}  (\bibinfo{year}{2014}).
\newblock \urlprefix\url{https://api.semanticscholar.org/CorpusID:118149905}.

\bibitem{PhysRevE.70.057701}
\bibinfo{author}{Marto\ifmmode~\check{n}\else \v{n}\fi{}\'ak, R.}, \bibinfo{author}{Santoro, G.~E.} \& \bibinfo{author}{Tosatti, E.}
\newblock \bibinfo{title}{Quantum annealing of the traveling-salesman problem}.
\newblock \emph{\bibinfo{journal}{Phys. Rev. E}} \textbf{\bibinfo{volume}{70}}, \bibinfo{pages}{057701} (\bibinfo{year}{2004}).
\newblock \urlprefix\url{https://link.aps.org/doi/10.1103/PhysRevE.70.057701}.

\bibitem{RevModPhys.80.1061}
\bibinfo{author}{Das, A.} \& \bibinfo{author}{Chakrabarti, B.~K.}
\newblock \bibinfo{title}{Colloquium: Quantum annealing and analog quantum computation}.
\newblock \emph{\bibinfo{journal}{Rev. Mod. Phys.}} \textbf{\bibinfo{volume}{80}}, \bibinfo{pages}{1061--1081} (\bibinfo{year}{2008}).
\newblock \urlprefix\url{https://link.aps.org/doi/10.1103/RevModPhys.80.1061}.

\bibitem{Lucas_2014}
\bibinfo{author}{Lucas, A.}
\newblock \bibinfo{title}{Ising formulations of many np problems}.
\newblock \emph{\bibinfo{journal}{Frontiers in Physics}} \textbf{\bibinfo{volume}{2}} (\bibinfo{year}{2014}).
\newblock \urlprefix\url{http://dx.doi.org/10.3389/fphy.2014.00005}.

\bibitem{ramezani2024reducingnumberqubitsn2}
\bibinfo{author}{Ramezani, M.}, \bibinfo{author}{Salami, S.}, \bibinfo{author}{Shokhmkar, M.}, \bibinfo{author}{Moradi, M.} \& \bibinfo{author}{Bahrampour, A.}
\newblock \bibinfo{title}{Reducing the number of qubits from $n^2$ to $n\log_{2} (n)$ to solve the traveling salesman problem with quantum computers: A proposal for demonstrating quantum supremacy in the nisq era} (\bibinfo{year}{2024}).
\newblock \urlprefix\url{https://arxiv.org/abs/2402.18530}.
\newblock \eprint{2402.18530}.

\bibitem{laporte1992traveling}
\bibinfo{author}{Laporte, G.}
\newblock \bibinfo{title}{The traveling salesman problem: An overview of exact and approximate algorithms}.
\newblock \emph{\bibinfo{journal}{European Journal of Operational Research}} \textbf{\bibinfo{volume}{59}}, \bibinfo{pages}{231--247} (\bibinfo{year}{1992}).

\bibitem{chauhan2012survey}
\bibinfo{author}{Chauhan, C.}, \bibinfo{author}{Gupta, R.} \& \bibinfo{author}{Pathak, K.}
\newblock \bibinfo{title}{Survey of methods of solving tsp along with its implementation using dynamic programming approach}.
\newblock \emph{\bibinfo{journal}{International journal of computer applications}} \textbf{\bibinfo{volume}{52}} (\bibinfo{year}{2012}).

\bibitem{christofides1976worst}
\bibinfo{author}{Christofides, N.}
\newblock \bibinfo{title}{Worst-case analysis of a new heuristic for the travelling salesman problem}.
\newblock \bibinfo{type}{Tech. Rep.}, \bibinfo{institution}{Carnegie-Mellon Univ Pittsburgh Pa Management Sciences Research Group} (\bibinfo{year}{1976}).

\bibitem{helsgaun2000effective}
\bibinfo{author}{Helsgaun, K.}
\newblock \bibinfo{title}{An effective implementation of the {L}in--{K}ernighan traveling salesman heuristic}.
\newblock \emph{\bibinfo{journal}{European Journal of Operational Research}} \textbf{\bibinfo{volume}{126}}, \bibinfo{pages}{106--130} (\bibinfo{year}{2000}).

\bibitem{johnson1990local}
\bibinfo{author}{Johnson, D.~S.}
\newblock \bibinfo{title}{Local optimization and the traveling salesman problem}.
\newblock In \emph{\bibinfo{booktitle}{International colloquium on automata, languages, and programming}}, \bibinfo{pages}{446--461} (\bibinfo{organization}{Springer}, \bibinfo{year}{1990}).

\bibitem{bengio2021machine}
\bibinfo{author}{Bengio, Y.}, \bibinfo{author}{Lodi, A.} \& \bibinfo{author}{Prouvost, A.}
\newblock \bibinfo{title}{Machine learning for combinatorial optimization: a methodological tour d’horizon}.
\newblock \emph{\bibinfo{journal}{European Journal of Operational Research}} \textbf{\bibinfo{volume}{290}}, \bibinfo{pages}{405--421} (\bibinfo{year}{2021}).

\bibitem{lombardi2018boosting}
\bibinfo{author}{Lombardi, M.} \& \bibinfo{author}{Milano, M.}
\newblock \bibinfo{title}{Boosting combinatorial problem modeling with machine learning}.
\newblock \emph{\bibinfo{journal}{arXiv preprint arXiv:1807.05517}}  (\bibinfo{year}{2018}).

\bibitem{doi:10.1137/1.9781611975482.107}
\bibinfo{author}{Ambainis, A.} \emph{et~al.}
\newblock \emph{\bibinfo{title}{Quantum Speedups for Exponential-Time Dynamic Programming Algorithms}}, \bibinfo{pages}{1783--1793}.
\newblock \urlprefix\url{https://epubs.siam.org/doi/abs/10.1137/1.9781611975482.107}.
\newblock \eprint{https://epubs.siam.org/doi/pdf/10.1137/1.9781611975482.107}.

\bibitem{Vargas_Calder_n_2021}
\bibinfo{author}{Vargas-Calderón, V.}, \bibinfo{author}{Parra-A., N.}, \bibinfo{author}{Vinck-Posada, H.} \& \bibinfo{author}{González, F.~A.}
\newblock \bibinfo{title}{Many-qudit representation for the travelling salesman problem optimisation}.
\newblock \emph{\bibinfo{journal}{Journal of the Physical Society of Japan}} \textbf{\bibinfo{volume}{90}}, \bibinfo{pages}{114002} (\bibinfo{year}{2021}).
\newblock \urlprefix\url{http://dx.doi.org/10.7566/JPSJ.90.114002}.

\bibitem{He_2024}
\bibinfo{author}{He, H.}
\newblock \emph{\bibinfo{title}{Quantum Annealing and GNN for Solving TSP with QUBO}}, \bibinfo{pages}{134–145} (\bibinfo{publisher}{Springer Nature Singapore}, \bibinfo{year}{2024}).
\newblock \urlprefix\url{http://dx.doi.org/10.1007/978-981-97-7801-0_12}.

\bibitem{Monta_ez_Barrera_2024}
\bibinfo{author}{Montañez-Barrera, J.~A.}, \bibinfo{author}{Willsch, D.}, \bibinfo{author}{Maldonado-Romo, A.} \& \bibinfo{author}{Michielsen, K.}
\newblock \bibinfo{title}{Unbalanced penalization: a new approach to encode inequality constraints of combinatorial problems for quantum optimization algorithms}.
\newblock \emph{\bibinfo{journal}{Quantum Science and Technology}} \textbf{\bibinfo{volume}{9}}, \bibinfo{pages}{025022} (\bibinfo{year}{2024}).
\newblock \urlprefix\url{http://dx.doi.org/10.1088/2058-9565/ad35e4}.

\bibitem{goldsmith2024qubohoboformulationssolving}
\bibinfo{author}{Goldsmith, D.} \& \bibinfo{author}{Day-Evans, J.}
\newblock \bibinfo{title}{Beyond qubo and hobo formulations, solving the travelling salesman problem on a quantum boson sampler} (\bibinfo{year}{2024}).
\newblock \urlprefix\url{https://arxiv.org/abs/2406.14252}.
\newblock \eprint{2406.14252}.

\bibitem{ali2024travelingsalesmanproblemtensor}
\bibinfo{author}{Ali, A.~M.}, \bibinfo{author}{Delgado, I.~P.} \& \bibinfo{author}{de~Leceta, A. M.~F.}
\newblock \bibinfo{title}{Traveling salesman problem from a tensor networks perspective} (\bibinfo{year}{2024}).
\newblock \urlprefix\url{https://arxiv.org/abs/2311.14344}.
\newblock \eprint{2311.14344}.

\bibitem{liu2024quantumlocalsearchtraveling}
\bibinfo{author}{Liu, C.-Y.}, \bibinfo{author}{Matsuyama, H.}, \bibinfo{author}{hao Huang, W.} \& \bibinfo{author}{Yamashiro, Y.}
\newblock \bibinfo{title}{Quantum local search for traveling salesman problem with path-slicing strategy} (\bibinfo{year}{2024}).
\newblock \urlprefix\url{https://arxiv.org/abs/2407.13616}.
\newblock \eprint{2407.13616}.

\bibitem{goswami2024solvingtravellingsalesmanproblem}
\bibinfo{author}{Goswami, K.}, \bibinfo{author}{Veereshi, G.~A.}, \bibinfo{author}{Schmelcher, P.} \& \bibinfo{author}{Mukherjee, R.}
\newblock \bibinfo{title}{Solving the travelling salesman problem using a single qubit} (\bibinfo{year}{2024}).
\newblock \urlprefix\url{https://arxiv.org/abs/2407.17207}.
\newblock \eprint{2407.17207}.

\bibitem{zhu2022realizablegasbasedquantumalgorithm}
\bibinfo{author}{Zhu, J.}, \bibinfo{author}{Gao, Y.}, \bibinfo{author}{Wang, H.}, \bibinfo{author}{Li, T.} \& \bibinfo{author}{Wu, H.}
\newblock \bibinfo{title}{A realizable gas-based quantum algorithm for traveling salesman problem} (\bibinfo{year}{2022}).
\newblock \urlprefix\url{https://arxiv.org/abs/2212.02735}.
\newblock \eprint{2212.02735}.

\bibitem{sato2024circuitdesigntwostepquantum}
\bibinfo{author}{Sato, R.} \emph{et~al.}
\newblock \bibinfo{title}{Circuit design of two-step quantum search algorithm for solving traveling salesman problems} (\bibinfo{year}{2024}).
\newblock \urlprefix\url{https://arxiv.org/abs/2405.07129}.
\newblock \eprint{2405.07129}.

\bibitem{Gilliam_2021CPBO}
\bibinfo{author}{Gilliam, A.}, \bibinfo{author}{Woerner, S.} \& \bibinfo{author}{Gonciulea, C.}
\newblock \bibinfo{title}{Grover adaptive search for constrained polynomial binary optimization}.
\newblock \emph{\bibinfo{journal}{Quantum}} \textbf{\bibinfo{volume}{5}}, \bibinfo{pages}{428} (\bibinfo{year}{2021}).
\newblock \urlprefix\url{http://dx.doi.org/10.22331/q-2021-04-08-428}.

\bibitem{chen2020maxORmin}
\bibinfo{author}{Chen, Y.} \emph{et~al.}
\newblock \bibinfo{title}{A low failure rate quantum algorithm for searching maximum or minimum}.
\newblock \emph{\bibinfo{journal}{Quantum Information Processing}} \textbf{\bibinfo{volume}{19}} (\bibinfo{year}{2020}).

\bibitem{durr1999quantumalgorithmfindingminimum}
\bibinfo{author}{Durr, C.} \& \bibinfo{author}{Hoyer, P.}
\newblock \bibinfo{title}{A quantum algorithm for finding the minimum} (\bibinfo{year}{1999}).
\newblock \urlprefix\url{https://arxiv.org/abs/quant-ph/9607014}.
\newblock \eprint{quant-ph/9607014}.

\bibitem{Boyer_1998}
\bibinfo{author}{Boyer, M.}, \bibinfo{author}{Brassard, G.}, \bibinfo{author}{Høyer, P.} \& \bibinfo{author}{Tapp, A.}
\newblock \bibinfo{title}{Tight bounds on quantum searching}.
\newblock \emph{\bibinfo{journal}{Fortschritte der Physik}} \textbf{\bibinfo{volume}{46}}, \bibinfo{pages}{493–505} (\bibinfo{year}{1998}).
\newblock \urlprefix\url{http://dx.doi.org/10.1002/(SICI)1521-3978(199806)46:4/5<493::AID-PROP493>3.0.CO;2-P}.

\bibitem{MYRVOLD2001281}
\bibinfo{author}{Myrvold, W.} \& \bibinfo{author}{Ruskey, F.}
\newblock \bibinfo{title}{Ranking and unranking permutations in linear time}.
\newblock \emph{\bibinfo{journal}{Information Processing Letters}} \textbf{\bibinfo{volume}{79}}, \bibinfo{pages}{281--284} (\bibinfo{year}{2001}).
\newblock \urlprefix\url{https://www.sciencedirect.com/science/article/pii/S0020019001001417}.

\bibitem{Marsh_2020}
\bibinfo{author}{Marsh, S.} \& \bibinfo{author}{Wang, J.~B.}
\newblock \bibinfo{title}{Combinatorial optimization via highly efficient quantum walks}.
\newblock \emph{\bibinfo{journal}{Physical Review Research}} \textbf{\bibinfo{volume}{2}} (\bibinfo{year}{2020}).
\newblock \urlprefix\url{http://dx.doi.org/10.1103/PhysRevResearch.2.023302}.

\bibitem{chiew2018graphcomparisonnonlinearquantum}
\bibinfo{author}{Chiew, M.}, \bibinfo{author}{de~Lacy, K.}, \bibinfo{author}{Yu, C.~H.}, \bibinfo{author}{Marsh, S.} \& \bibinfo{author}{Wang, J.~B.}
\newblock \bibinfo{title}{Graph comparison via nonlinear quantum search} (\bibinfo{year}{2018}).
\newblock \urlprefix\url{https://arxiv.org/abs/1810.01647}.
\newblock \eprint{1810.01647}.

\bibitem{akiyama1980np}
\bibinfo{author}{Akiyama, T.}, \bibinfo{author}{Nishizeki, T.} \& \bibinfo{author}{Saito, N.}
\newblock \bibinfo{title}{{NP}-completeness of the {H}amiltonian cycle problem for bipartite graphs}.
\newblock \emph{\bibinfo{journal}{Journal of Information processing}} \textbf{\bibinfo{volume}{3}}, \bibinfo{pages}{73--76} (\bibinfo{year}{1980}).

\bibitem{Brassard_2002}
\bibinfo{author}{Brassard, G.}, \bibinfo{author}{Høyer, P.}, \bibinfo{author}{Mosca, M.} \& \bibinfo{author}{Tapp, A.}
\newblock \bibinfo{title}{Quantum amplitude amplification and estimation} (\bibinfo{year}{2002}).
\newblock \urlprefix\url{http://dx.doi.org/10.1090/conm/305/05215}.

\bibitem{Long_2001}
\bibinfo{author}{Long, G.~L.}
\newblock \bibinfo{title}{Grover algorithm with zero theoretical failure rate}.
\newblock \emph{\bibinfo{journal}{Physical Review A}} \textbf{\bibinfo{volume}{64}} (\bibinfo{year}{2001}).
\newblock \urlprefix\url{http://dx.doi.org/10.1103/PhysRevA.64.022307}.

\end{thebibliography}

\pagebreak
\newpage
\onecolumngrid
\beginsupplement
\setcounter{figure}{0}
\renewcommand{\thefigure}{S\arabic{figure}}

\section*{Supplementary Numerical Results}

    \begin{figure}[H]
        \centering
        \includegraphics[width=1.\linewidth]{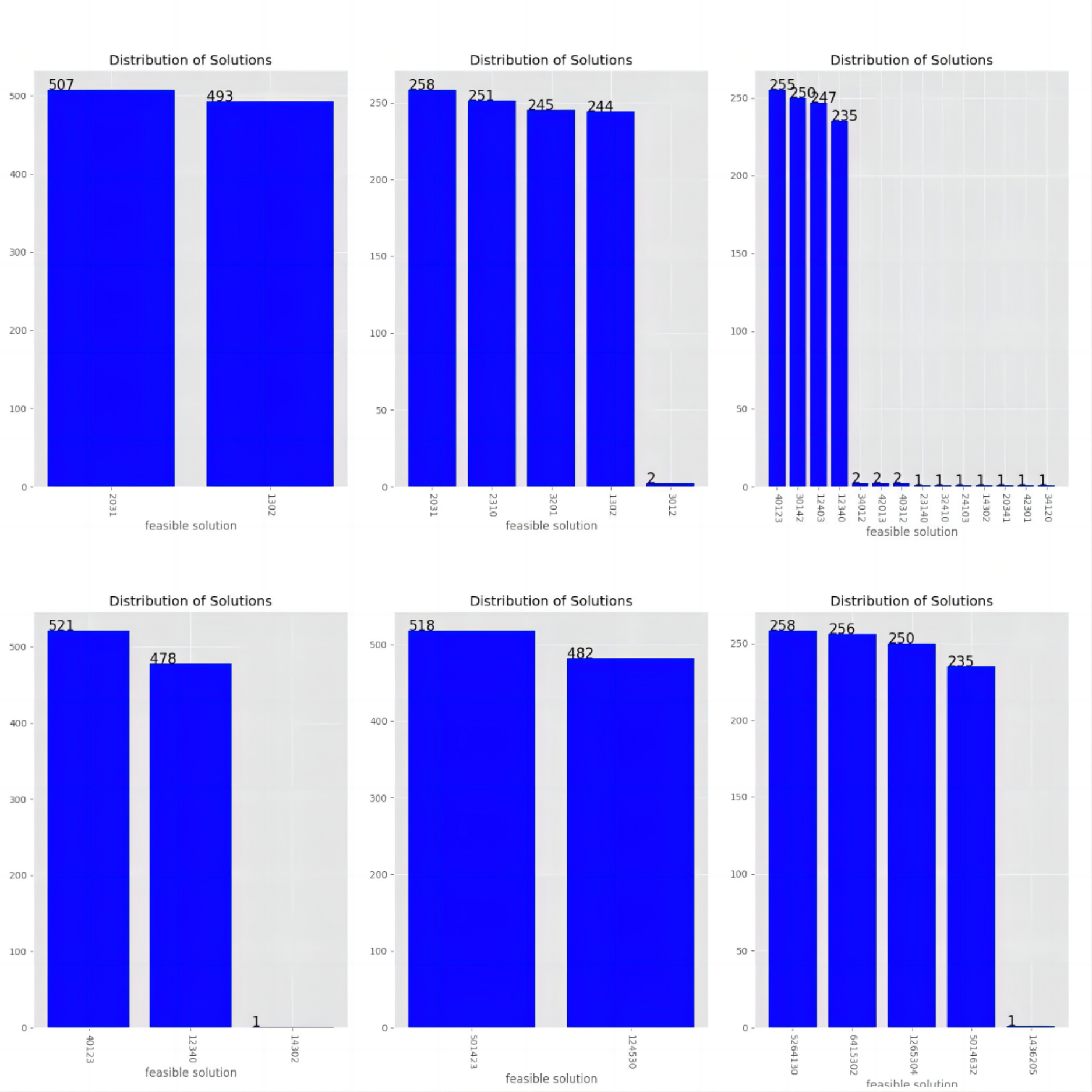}

        \caption{Simulation results of 1000 samples for each of six graphs, corresponding to $X_1 \sim X_6$ in {\bf Table~\ref{tab:sim}} from left to right and top to bottom. The bars with significant proportions are all optimal solutions. Each feasible solution is a permutation. For example, $\sigma=(2031)$ represents the route $0\to2\to3\to1\to0$.}
        \label{app:graph-results}
    \end{figure}

    \begin{figure}[H]
        \centering
        \includegraphics[width=1.\linewidth]{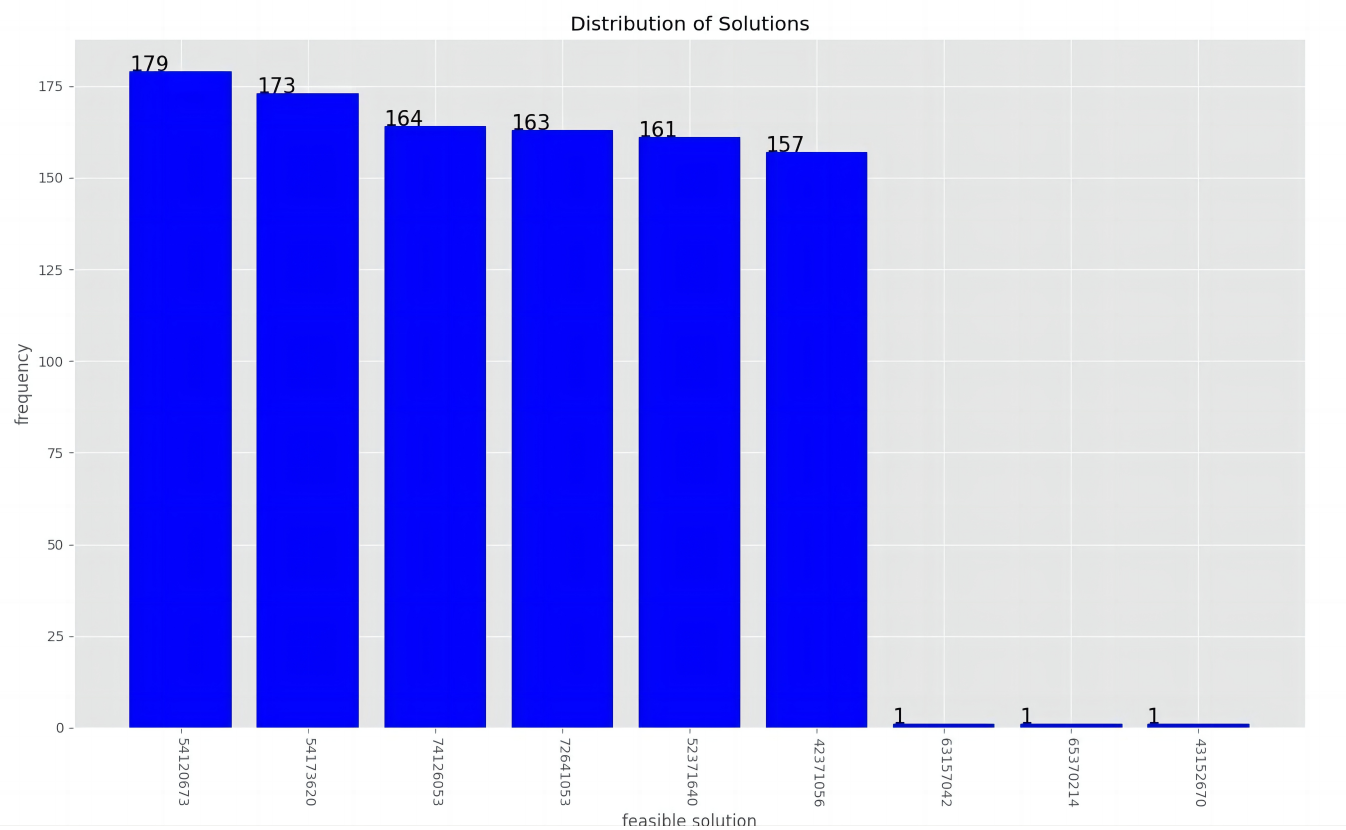}

        \caption{Simulation results of 1000 samples for eight-node TSP example, corresponding to $X_7$ in {\bf Table~\ref{tab:sim}}. The bars with significant proportions are all optimal solutions.}
        \label{app:graph-results1}
    \end{figure}

    \begin{figure}[H]
        \centering
        \includegraphics[width=1.\linewidth]{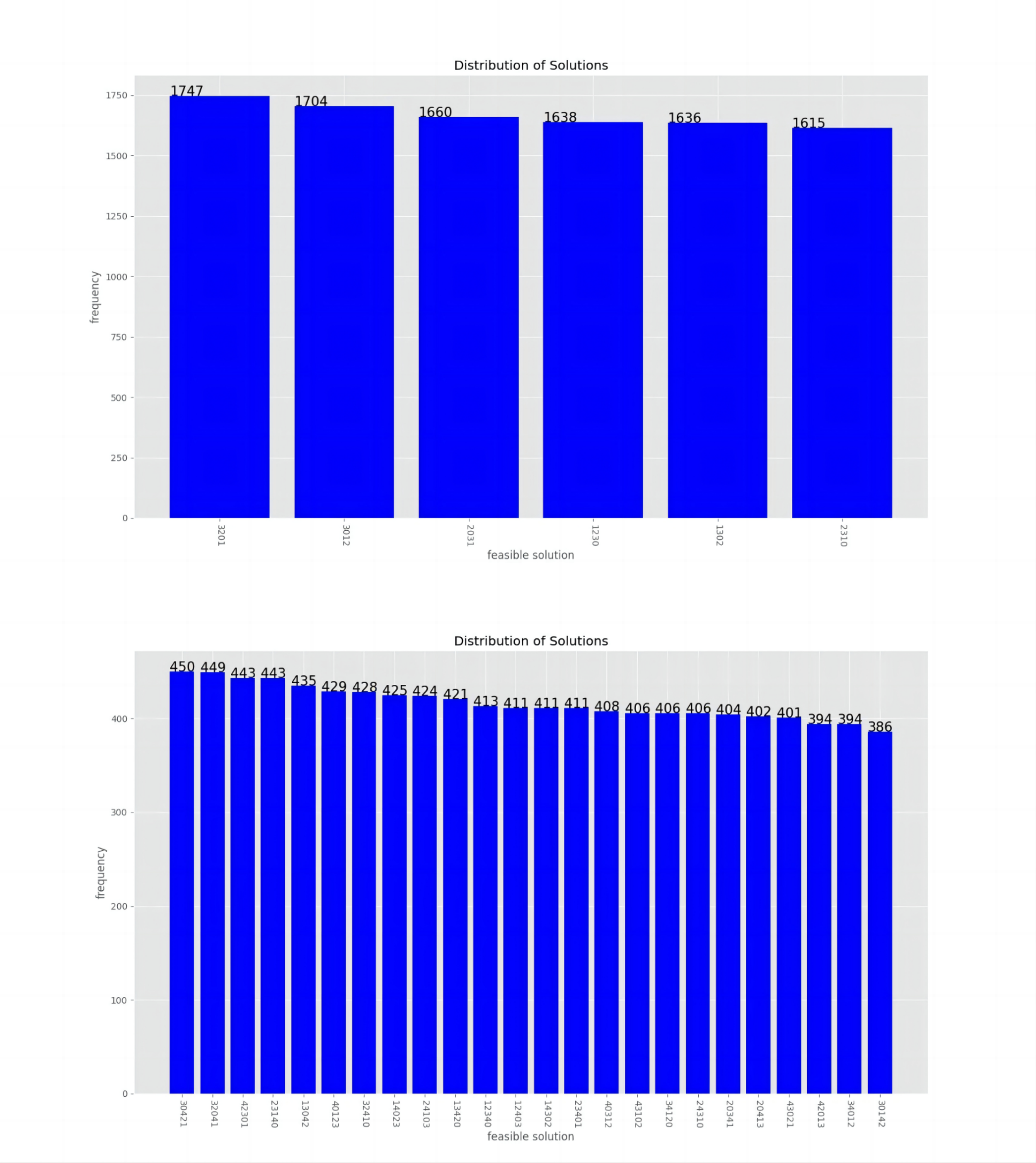}

        \caption{Simulation results of QDP HC-generation algorithm on 4-length and 5-length HC, where sampling size is 10000. Each non-zero frequency bar represents a Hamiltonian Cycle, which means there are no infeasible states.}
        \label{app:graph-results2}
    \end{figure}

\end{document}